\documentclass[11pt]{llncs}
\usepackage[english]{babel}
\usepackage[T1]{fontenc}
\usepackage{amsmath}
\usepackage{amssymb}
\usepackage{amscd}
\usepackage{algorithm2e}
\usepackage{enumerate}
\usepackage{array}
\usepackage{extarrows}

\usepackage{graphicx} 

\usepackage{pgf}
\usepackage{tikz}
\usepackage{tkz-graph}
\usetikzlibrary{matrix,arrows}

\definecolor{emphcol}{rgb}{0.86,0.125,0.07}

\newif\ifcomment
\commenttrue
\definecolor{gris}{gray}{0.3}

\newcommand{\cqfd}[0]{~\hfill$\square$}

\newcommand{\tuple}[1]{{\langle #1 \rangle}}

\newcommand{\fun}[1]{\mathit{#1}}
\newcommand{\Nset}[0]{\bbbn}

\newcommand{\partset}[1]{\ensuremath{{\cal P}}(#1)} %% ensemble des parties
\newcommand{\eqdef}[0]{\; \triangleq \; } %% equivalence par d?finition
\newcommand{\abbrev}[1]{#1, \relax}
\newcommand{\ie}[0]{\abbrev{\textit{i.e.}}}
\newcommand{\eg}[0]{\abbrev{\textit{e.g.}}}

\newcommand{\restrict}[2]{\left.{#1}\right |_{#2}}
\newcommand{\diff}[0]{\setminus}
\newcommand{\lset}[0]{\Psi}
\newcommand{\LsetFun}[1]{\lset_{#1}}
\newcommand{\Lset}[2]{\LsetFun{#1}({#2})}

\newcommand{\orbit}[2]{\Omega_{#1}({#2})}
\newcommand{\orbitFun}[1]{\Omega_{#1}}
\newcommand{\evo}[0]{\rightharpoonup}

\newcommand{\nevo}[0]{\not \evo} %%\nrightharpoonup}

\newcommand{\reg}[0]{\longrightarrow}

\newcommand{\acd}[0]{\ensuremath{\tuple{A,S,(\evo_a)_{a \in A}}}}

\newcommand{\eqr}[1]{\,\sim_{#1}\,}

\begin{document}

\title{Modular organisation of interaction networks based on
  asymptotic dynamics}

\author{Franck Delaplace\inst{1,}\thanks{Corresponding author:
    \url{franck.delaplace@ibisc.univ-evry.fr}.} \and Hanna Klaudel\inst{1} \and
  Tarek Melliti\inst{1} \and Sylvain Sen{\'e}\inst{1,2}}

\institute{Universit{\'e} d'{\'E}vry -- Val d'Essonne, IBISC, {\'E}A 4526, 91000
  {\'E}vry, France\and Institut rh{\^o}ne-alpin des syst{\`e}mes complexes,
  IXXI, 69007 Lyon, France}

\date{}

\maketitle

\begin{abstract}
  This paper investigates questions related to the modularity in discrete models
  of biological interaction networks. We develop a theoretical framework based
  on the analysis of their asymptotic dynamics. More precisely, we exhibit
  formal conditions under which agents of interaction networks can be grouped
  into modules. As a main result, we show that the usual decomposition in
  strongly connected components fulfils the conditions of being a modular
  organisation. Furthermore, we point out that our framework enables a finer
  analysis providing a decomposition in elementary modules.
\end{abstract}

\section{Introduction}
\label{sec:introduction}

Understanding and exhibiting the relations between phenotypes and interactions
in biological networks, \ie the links between structures and
functions~\cite{Monod1970}, are among the most challenging problems at the
frontier of theoretical computer science and biology. Many phenotypes can be
associated to interactions of biological agents (\eg genes or proteins) working
together to guarantee some specific functions. This leads to group agents into
modules and to associate to them one or more biological functions. It follows
that biological networks can be seen, at a more abstract level, as modular
networks in which interacting elements are modules that carry the biological
information necessary to translate into precise functions.\bigskip

\noindent Modularity is present in various kinds of networks including metabolic
or signalling pathways and genetic or protein interaction networks. Modular
organisations are notably emphasised in
embryogenesis~\cite{Peter2009,McDougall2011} where modules of genes are
coordinated in the development process. Furthermore, methods related to modules
discovery in interaction networks are generally based on both the analysis of
the networks structures (a field close to graph theory) and the study of their
associated dynamics~\cite{Qi2006}. \emph{Structural analysis} identifies
sub-networks with specific topological properties motivated either by a
correspondence between topology and
functionality~\cite{Gagneur2004,Chaouiya2011} or by the existence of statistical
biases with respect to random networks~\cite{Rives2003}. Specific topologies
like cliques~\cite{Spirin2003}, or more generally strongly connected components
(SCCs) are commonly used to reveal modules by structural analysis. Particular
motifs~\cite{Milo2002,Alon2003} may also be interpreted as modules viewed as
basic components. They represent over-expressed biological sub-networks with
respect to random ones. \emph{Dynamical analysis} lays on the hypothesis that
expression profiles provide insights on the relationships between regulators,
modules being possibly revealed from correlations between the biological agents
expressions. For instance, using yeast gene expression data, the authors
of~\cite{Bar2003,Segal2003} inferred models of co-regulated genes and the
condition under which the regulation occurs. As a consequence, the discovery of
a modular organisation in biological interaction networks is closely related to
the influence of the expression of agents on the others and needs to investigate
their expression
dynamics~\cite{Thieffry1999,Han2008,Siebert2009,Delaplace2010,Demongeot2010}.\smallskip

\noindent Using a discrete model of biological interaction
networks~\cite{Thomas1973,Thomas1981}, we propose an approach that analyses the
conditions of modules formation and characterises the relations between the
global behaviour of a network and the local behaviours of its components.  The
chosen model of interaction networks is based on the assumption that phenotypes
are related, at the molecular level, to characteristic states of some biological
agents assimilated to equilibria\footnote{In systems biology, equilibria are
  usually associated to differentiated cellular types~\cite{Delbruck1949} and
  specific behavioural properties like lysis and lysogeny in the bacteriophage
  $\lambda$~\cite{Thieffry1995}.}. Thus, using dynamical analysis, we show
conditions under which interaction networks can be divided into modules so that
the composition of these modules behaviours matches the global behaviours of the
networks.\smallskip

\noindent The paper is structured as follows:
First, Section~\ref{sec:preliminaries} introduces the main definitions
and notations used throughout the paper. Then, Section~\ref{sec:composition}
presents the central notion of a modular organisation of a network along with
its structural and dynamical properties. Section~\ref{sec:separation} defines
elementary modular organisation and the conditions leading to obtain it. Some
concluding remarks and perspectives are provided in
Section~\ref{sec:discussion}.

\section{Preliminaries}
\label{sec:preliminaries}

First, we introduce basic notations. Let $\evo \; \subseteq S \times S$ be a
binary relation on $S$. Given $s, s'\in S$ and $S'\subseteq S$, we denote by $s
\evo s'$ the fact that $(s,s') \in\ \evo$, by $\left (s \, \evo \right) \eqdef
\{ s' \mid s \evo s'\}$ the image of $s$ by $\evo$, and by $\left (S' \evo
\right )$ its generalisation to $S'$. Similarly, we denote by $\left ( \evo \, s
\right )$ and $\left ( \evo S' \right )$ the corresponding preimages. The
composition of two binary relations will be denoted by $\evo \circ \evo'$, the
reflexive and transitive closure by $\evo^*$. 

\subsection{Interaction network and its associated dynamics}
\label{sec:global_dynamics}

Let us assume that a network $\eta$ is composed of a set $A = \{a_1, \ldots,
a_n\}$ of agents.  Each $a_i \in A$ has a \emph{local state}, denoted by
$s_{a_i}$, taking values in some nonempty finite set $S_{a_i}$. A \emph{state}
(or a configuration) of $\eta$ is defined as a vector $s\in S$ associating to
each $a_i\in A$ a value in $S_{a_i}$, where $S\eqdef S_{a_1}\times\ldots\times
S_{a_n}$ is the set of all possible states of the network. For any $X \subseteq
A$ and $s\in S$, we denote by $\restrict{s}{X}$ the restriction of $s$ to the
agents in $X$; this notation naturally extends to sets of states.\smallskip

\noindent An \emph{evolution} of $\eta$ is a relation $\evo \; \subseteq S
\times S$ where each $s \evo s'$ is a \emph{transition} meaning that $s$ evolves
to $s'$ by $\evo$. Thus, the \emph{global evolution} of $\eta$ can be
represented by a directed graph ${\cal G} = (S, \evo)$ called the \emph{state
  graph} (or the transition graph). In this work, we pay particular attention to
the notion of \emph{local evolution}, since each agent $a\in A$ has its own
evolution $\evo_a$. The collection of all these local evolutions results in the
\emph{asynchronous} view of the global evolution of $\eta$.
\begin{definition}
  \label{def:local_dyn}
  The \emph{asynchronous dynamics} (or \emph{dynamics} for short) of a network
  $\eta$ is defined as the triple $\acd$, where $A$ is a set of agents, $S$ is a
  set of states, and for each $a \in A$, $\evo_a \, \subseteq S \times S$ is a
  total or empty relation characterising the evolution of agent $a$ such that
  for any $s\evo_a s'$, either $s=s'$ or $s$ differs from $s'$ only on the
  $a$-th component.
\end{definition}
We are now in a position to introduce formally the \emph{interaction network} as
a family of functions\footnote{Actually, $\eta:A \to S \to \bigcup_{a \in A}
  S_a.$} $\eta = \{ \eta_a \}_{a \in A}$, such that each $\eta_a: S \to S_a$
defines the next state $\eta_a(s)$ with respect to the asynchronous evolution of
$a$ from $s$. Network $\eta$ allows to deduce a directed graph of interactions
$G \eqdef (A,\reg)$ such that $a_i \reg a_j$ if $a_i$ occurs in the definition
of $\eta_{a_j}$. Notice that $\eta$ implies $G$, but the opposite is in general
not true. When $\evo_a$ is empty for some $a$ (\ie the local state of $a$
remains invariant), then $a$ plays the role of an input, which means that no
other agents of $A$ influences it (\ie there are no arcs towards $a$ in
$G$).\smallskip

\noindent Given a set $S'\subseteq S$, we introduce the notions of orbits,
equilibria and attractors:\\[-7mm]
\begin{itemize}
\item An \emph{orbit} of $S'$ is the set of states $\orbit{}{S'} \eqdef (S'
  \evo^*)$ comprising $S'$ and all the states reachable from $S'$;
\item An \emph{equilibrium} $e\in S$ is a state endlessly reachable by $\evo$;
  $\Lset{}{S'}$ denotes the set of equilibria reachable from $S'$;
\item An \emph{attractor} is a set of equilibria $E \subseteq S$ such that
  $\forall e \in E \colon \Lset{}{\{e\}}=E$. In a state graph, attractors are
  sets of states belonging to terminal strongly connected components that can
  be of two kinds:\\
  -- a \emph{stable state} is a singleton $E\subseteq S$;\\
  -- a \emph{limit set} is an attractor $E$ such that $|E|>1$.
\end{itemize}

\noindent The restrictions of $\evo$ to $X\subseteq A$, denoted by $\evo_X =
\bigcup_{a \in X} \evo_a$, and of a state $s$, denoted by $s|_X$, lead to the
following definitions.
\begin{definition}
  \label{def:eqr}
  Two states $s_1, s_2 \in S$ are said to be \emph{X-equivalent} and are denoted
  by $s_1 \eqr{X} s_2$, for some $X\subseteq A$, if $\restrict{s_1}{X} =
  \restrict{s_2}{X}$, \ie if they are the same for all agents in $X$.
\end{definition}
This equivalence relation naturally extends to sets of agents, \ie $S_1 \eqr{X}
S_2 \eqdef \restrict{S_1}{X} = \restrict{S_2}{X}$.
\begin{definition}
  \label{def:equilibrium-orbit}
  We consider the following operators with the same signature $\partset{A} \to
  (\partset{S} \to \partset{S})$, where $X \subseteq A$ and $S' \subseteq 
  S$:\\[-7mm]
  \begin{itemize}
  \item the \emph{orbit operator} $\Omega$, defined as $\orbit{X}{S'} =
    (S'\evo_X^*)$,
  \item the \emph{equilibria operator} $\Psi$, defined as $\Lset{X}{S'} = \{ s
    \in \orbit{X}{S'} \mid \forall s' \in S: s \evo_X^* s' \implies s' \evo_X^*
    s \}.$
  \end{itemize}
\end{definition}
Proposition~\ref{prop:eq} below emphasises specific properties of operator
$\Psi$.
\begin{proposition} 
  \label{prop:eq}
  Let $X \subseteq A$ be a subset of agents. $\LsetFun{X}$ has the following
  properties, for all sets $S',S''$ subsets of states:
  \begin{enumerate}[a.]
  \item Idempotency: $\LsetFun{X} \circ \LsetFun{X}(S') = \Lset{X}{\Lset{X}{S'}}
    = \Lset{X}{S'}$; \label{idempotency}
  \item Upper-continuity: $\Lset{X}{S' \cup S''} = \Lset{X}{S'} \cup
    \Lset{X}{S''}$; \label{eq3}
  \item Monotony (order-preserving) : $S' \subseteq S'' \implies \Lset{X}{S'}
    \subseteq \Lset{X}{S''}$. \label{monotony}
  \end{enumerate}
\end{proposition}

\begin{proof}
  Let $\fun{Eq}_{X}(s)$ be the predicate meaning that $s$ is an equilibrium for
  $\evo_{X}$.
  \begin{enumerate}[a.]
  \item By expanding $\Lset{X}{\Lset{X}{S'}}$, we have:
    \begin{equation*}
      \begin{split}
        \Lset{X}{\Lset{X}{S'}} & = \{s \in \orbit{X}{\Lset{X}{S'}} \mid
        \fun{Eq}_{X}(s)\}\\
        & = \{s \in \Lset{X}{S'} \mid \fun{Eq}_{X}(s)\}\\
        & = \Lset{X}{S'}\text{.}
      \end{split}
    \end{equation*}
  \item By definition, $\Lset{X}{S' \cup S''} = \{ s \in \orbit{X}{S' \cup S''}
    \mid \fun{Eq}_X(s) \}$, where $\orbit{X}{S' \cup S''} = (S' \cup
    S'')\hspace*{-3pt}\evo_X^*$. Since $\evo_X^*$ is upper-continuous on the
    lattice of state sets, we have:
    \begin{equation*}
      \begin{split}
        \Lset{X}{S' \cup S''} & = \{s \in \orbit{X}{S' } \cup \orbit{X}{S''}
        \mid \fun{Eq}_X(s)\}\\
        & = \{s \in \orbit{X}{S' } \mid \fun{Eq}(s) \} \cup \{ s \in
        \orbit{X}{S''} \mid \fun{Eq}_X(s)\}\\
        & = \Lset{X}{S'} \cup \Lset{X}{S''}\text{.}
      \end{split}
    \end{equation*}
  \item An upper-continuous function is monotonous.\cqfd
  \end{enumerate}
\end{proof}
\begin{figure}[t!]
  \begin{center}
    \begin{tabular}{ccm{56.3mm}}
      $\eta=\left\{
        \begin{array}{l}
          \eta_{a_1}(s) = s_{a_1}\\
          \eta_{a_2}(s) = s_{a_1} \lor s_{a_3}\\
          \eta_{a_3}(s) = \neg s_{a_2}\\
          \eta_{a_4}(s) = s_{a_3}\\
        \end{array}
      \right. $ & ~~~~~~ &
      \includegraphics[scale=1]{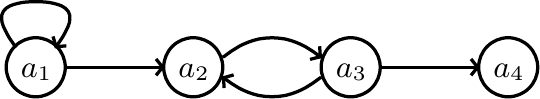}
      % \begin{tikzpicture}[scale=1, node distance=1.5cm]
      %   \SetVertexNormal[Shape = circle, LineWidth=0.5pt]
      %   \Vertices[Math, x=0, y=0, dir=\EA]{a_1,a_2,a_3,a_4};
      %   \tikzset{EdgeStyle/.style = {post}}
      %   \draw (a_1) to [->, loop left, looseness=0, min distance=8mm, in=45, 
      %   out=135](a_1);
      %   \draw (a_1) to [->] (a_2);
      %   \draw (a_2) to [->, bend left] (a_3);
      %   \draw (a_3) to [->, bend left] (a_2);
      %   \draw (a_3) to [->] (a_4);
      % \end{tikzpicture}
    \end{tabular}\vspace*{4mm}
    \hspace*{3pt}\begin{tikzpicture}[xscale=2.9, yscale=1.5]
      \scriptsize %% size of node labels
      \GraphInit[vstyle=Normal] %% type of graph
      \SetVertexNormal[Shape = rectangle, LineWidth=0pt]
      \node[rectangle, draw] at (0,0) (1000) {$\scriptstyle 1000$};
      \node[rectangle, draw] at (0,1) (1001) {$\scriptstyle 1001$};
      \node[rectangle, draw] at (1,0) (1010) {$\scriptstyle 1010$};
      \node[rectangle, draw] at (1,1) (1011) {$\scriptstyle 1011$};
      \node[rectangle, draw, fill=black!20] at (.7,.25) (1100) 
      {$\scriptstyle 1100$};
      \node[rectangle, draw] at (.7,1.25) (1101) {$\scriptstyle 1101$};
      \node[rectangle, draw] at (1.7,.25) (1110) {$\scriptstyle 1110$};
      \node[rectangle, draw] at (1.7,1.25) (1111) {$\scriptstyle 1111$};
      \node[rectangle, draw, color=white, fill=black!65] at (2.1,0) (0000) 
      {$\scriptstyle 0000$};
      \node[rectangle, draw, color=white, fill=black!65] at (2.1,1) (0001) 
      {$\scriptstyle 0001$};
      \node[rectangle, draw, color=white, fill=black!65] at (3.1,0) (0010) 
      {$\scriptstyle 0010$};
      \node[rectangle, draw, color=white, fill=black!65] at (3.1,1) (0011) 
      {$\scriptstyle 0011$};
      \node[rectangle, draw, color=white, fill=black!65] at (2.8,0.25) (0100)
      {$\scriptstyle 0100$};
      \node[rectangle, draw, color=white, fill=black!65] at (2.8,1.25) (0101) 
      {$\scriptstyle 0101$};
      \node[rectangle, draw, color=white, fill=black!65] at (3.8,0.25) (0110)
      {$\scriptstyle 0110$};
      \node[rectangle, draw, color=white, fill=black!65] at (3.8,1.25) (0111) 
      {$\scriptstyle 0111$};
      %% evolutions
      \tikzset{EdgeStyle/.style = {post}}
      \tikzset{LabelStyle/.style = {fill=white,sloped}}
      \tikzset{LabelStyle/.style = {color=white, text=black}}
      % arrow of a_2
      \Edge[label={$a_2$}](0100)(0000);
      \Edge[label={$a_2$}](0010)(0110);
      \Edge[label={$a_2$}](1000)(1100);
      \Edge[label={$a_2$}](1010)(1110);
      \Edge[label={$a_2$}](0101)(0001);
      \Edge[label={$a_2$}](0011)(0111);
      \Edge[label={$a_2$}](1001)(1101);
      \Edge[label={$a_2$}](1011)(1111);
      % arrow of a_3
      \Edge[label={$a_3$}](0110)(0100);
      \Edge[label={$a_3$}](0111)(0101);
      \Edge[label={$a_3$}](1110)(1100);
      \Edge[label={$a_3$}](1111)(1101);
      \Edge[label={$a_3$}](0000)(0010);
      \Edge[label={$a_3$}](0001)(0011);
      \Edge[label={$a_3$}](1000)(1010);
      \Edge[label={$a_3$}](1001)(1011);
      %% arrow of a_4
      \Edge[label={$a_4$}](0001)(0000);
      \Edge[label={$a_4$}](0101)(0100);
      \Edge[label={$a_4$}](1001)(1000);
      \Edge[label={$a_4$}](1101)(1100);
      \Edge[label={$a_4$}](0010)(0011);
      \Edge[label={$a_4$}](0110)(0111);
      \Edge[label={$a_4$}](1010)(1011);
      \Edge[label={$a_4$}](1110)(1111);
    \end{tikzpicture}
    \caption{An interaction network and its state graph.}
    \label{fig:examplea}
  \end{center}
\end{figure} 

% \subsection{Example of the dynamics of a Boolean network}
% \label{sec:example}

\noindent Example in Figure~\ref{fig:examplea} illustrates the dynamics of a
Boolean network\footnote{In the state graphs of interaction networks cited as
  examples, self loops are omitted. Stable states are depicted in grey while
  limit sets are in black.}, \ie in which for all $a \in A$, $S_a = \{0,1\}$. It
is defined by $\eta$, each $\eta_a$ being the local transition function of agent
$a$. Given a state $s$, the evolution $s \evo_a s'$ means that $s'$ is obtained
by applying $\eta_a$ to $s$. More formally:
\begin{equation*}
  s \evo_a s' \eqdef s'_a = \eta_a(s) \land \left( \forall a' \in A \diff \{a\}:
    s'_a = s_a \right)\text{.}
\end{equation*}
Figure~\ref{fig:examplea} (top) provides $\eta$ and its graphical
representation. Figure~\ref{fig:examplea} (bottom) depicts the state graph
${\cal G}$ of $\eta$. More precisely, ${\cal G}$ represents the dynamics of
$\eta$ for each state $s \in \{0,1\}^4$. Let us remark that:\\[-6mm]
\begin{itemize}
\item the \emph{orbit} of $\{1111\}$ is $\orbit{}{\{1111\}}=\{1111, 1101,
  1100\}$;
\item $1100$ is an \emph{equilibrium}, as well as $0101$ and $0010$ are;
\item the set of equilibria reachable from $0000$ is $\{0xyz\mid
  x,y,z\in \{0,1\}\}$;
\item two attractors exist: a stable state $\{1100\}$ and a limit set
  $\Lset{}{\{0000\}}$.
\end{itemize}

\subsection{Regulation}
\label{sec:regulation}

The regulation relation specifies a state-based dependence between two
agents. Agent $a_k$ is a \emph{regulator} of agent $a_{\ell}$, denoted by $a_k
\reg a_{\ell}$, if at least one modification of the state of $a_{\ell}$ requires
a modification of the state of $a_k$.
\begin{definition}
  \label{def:reg}
  \emph{$a_k$ regulates $a_{\ell}$}, denoted by $a_k \reg a_{\ell}$, if and only
  if there exist two states $s,s' \in S$ such that $(s \eqr{A \diff \{a_k\}} s')
  \land ((s \evo_{a_{\ell}}) \nsim_{a_{\ell}} (s' \evo_{a_{\ell}}))$. By
  extension, given $X_i, X_j \subseteq A$, $X_i \reg X_j$ if and only if
  $\exists a_k \in X_i, \exists a_{\ell} \in X_j: a_k \reg a_{\ell}$.
\end{definition}
In Figure~\ref{fig:examplea}, the sets of regulators of agents $a_1$, $a_2$,
$a_3$ and $a_4$ are respectively $\{a_1\}$, $\{a_1, a_3\}$, $\{a_2\}$ and
$\{a_3\}$. Notice also that there are the following relations on sets of agents:
$\{a_1,a_2\} \reg \{a_3,a_4\}, \{a_3\} \reg \{a_2,a_4\}$ and $\{a_1\} \reg
\{a_2,a_3\} \reg\{a_4\}$. Another example of a regulation graph is given in
Figure~\ref{fig:exampleb} (right). It shows that interaction $(a_1,a_2)$ in the
interaction network is actually not a regulation because no modification of
$a_1$ influences the state of $a_2$. All the other interactions are effective,
meaning that the underlying regulation graph contains all interactions from the
network but $(a_1,a_2)$.
\begin{figure}[t!]
  \begin{center}
    \hspace*{-3pt}\begin{minipage}{45mm}
      \begin{tabular}{c}
          $\eta=\left\{
            \begin{array}{l}
              \eta_{a_1}(s) = s_{a_1} \land s_{a_2}\\
              \eta_{a_2}(s) = s_{a_1} \lor 1\\
            \end{array}
          \right.$\\[5mm]
        % \begin{tikzpicture}[scale=1, node distance=1.5cm]
        %   \SetVertexNormal[Shape=circle, LineWidth=0.5pt] 
        %   \Vertices[Math, x=0, y=0, dir=\EA]{a_1,a_2};
        %   \tikzset{EdgeStyle/.style = {post}}
        %   \draw (a_1) to [->,loop left, looseness=0, min distance=8mm, in=210, 
        %   out=150](a_1);
        %   \draw (a_1) to [->,bend left] (a_2);
        %   \draw (a_2) to [->, bend left] (a_1);
        % \end{tikzpicture}
          \includegraphics[scale=1]{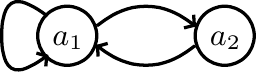}
      \end{tabular}
    \end{minipage}
    \begin{minipage}{35.3mm}
      \vspace*{-1.5mm}\begin{tikzpicture}[xscale=2.9, yscale=1.5]% scale=2.5
        \scriptsize %% size of node labels
        \GraphInit[vstyle=Normal] %% type of graph 
        \SetVertexNormal[Shape = rectangle, LineWidth=0.5pt] 
        \node[rectangle,draw] at (0,0) (00) {$\scriptstyle 00$};
        \node[rectangle,draw,fill=black!20] at (0,0.75) (01) 
        {$\scriptstyle 01$};
        \node[rectangle,draw] at (0.75,0) (10) {$\scriptstyle 10$};
        \node[rectangle,draw,fill=black!20] at (0.75,0.75) (11) 
        {$\scriptstyle 11$};
        %% evolutions 
        \tikzset{EdgeStyle/.style = {post}}
        \tikzset{LabelStyle/.style = {fill=white,sloped}}
        \tikzset{LabelStyle/.style = {color=white, text=black}}
        % arrow of a_1
        \Edge[label={$a_1$}](10)(00);
        % arrow of a_2
        \Edge[label={$a_2$}](00)(01);
        \Edge[label={$a_2$}](10)(11);
      \end{tikzpicture}
    \end{minipage}
    \begin{minipage}{27.8mm}
      % \begin{tikzpicture}[scale=1, node distance=1.5cm]
      %   \SetVertexNormal[Shape=circle, LineWidth=0.5pt] 
      %   \Vertices[Math, x=0, y=0, dir=\EA]{a_1,a_2};
      %   \tikzset{EdgeStyle/.style = {post}}
      %   \draw (a_1) to [->, loop left, looseness=0, min distance=8mm, in=210,
      %   out=150](a_1);
      %   \draw (a_2) to [->, bend left] (a_1);
      % \end{tikzpicture}
      \centerline{\includegraphics[scale=1]{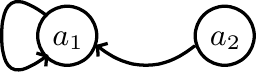}}
    \end{minipage}
  \end{center}
  \vspace*{-7mm}\caption{An interaction network, its state graph
    and the corresponding regulation graph.}
  \label{fig:exampleb}
\end{figure}
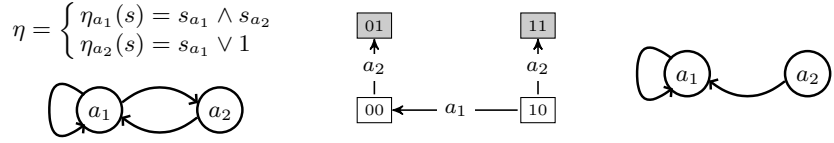

\section{Composition of equilibria }
\label{sec:composition}

In this section, we give rise to the relation between the equilibria of an
interaction network and that of its parts. This allows to consider a modular
view of the system in which each part is assimilated to a \emph{module}, \ie a
subset of agents with their dynamics. It means that modules, which influence
each other, reveal the underlying \emph{biological functions} materialised by
their equilibria.

\subsection{Modular organisation}
\label{sec:modular-organisation}

Our objective is to find a decomposition of $A$ into a set of modules, \ie a
partition\footnote{A \emph{partition} of a set $A$ is a set of nonempty disjoint
  subsets of $A$ which covers $A$.} of $A$, together with a composition operator
$\oslash$ of the equilibria of these modules, allowing to retrieve the global
equilibria of the network.  It is actually a challenge to find an adequate
operator $\oslash$, \ie such that a decomposition into a set of modules $\{X_1,
\ldots, X_m\}$ satisfies:
\begin{equation}
  \label{eq:compo_part}
  \LsetFun{X_1} \oslash \ldots \oslash \LsetFun{X_m} = 
  \LsetFun{\bigcup_{i=1}^m X_i}\text{.}
\end{equation}
One can easily see that, in general, taking $\oslash=\cup$ is not a
solution. For example, if we consider the interaction network in
Figure~\ref{fig:examplec} and a partition into two modules $\{a_1,a_2\}$ and
$\{a_3\}$, then the corresponding sets of equilibria are respectively
$\{110,111\}$ and $\{000,001,100,101,011,111\}$, while the set of global
equilibria is $\{111\}$, which is not the union of the previous ones.  However,
one may see that, for the same modules, the computation of the equilibria of
$\{a_3\}$ from the equilibria of $\{a_1,a_2\}$, gives the desired property
$\LsetFun{\{a_3\}} \circ \LsetFun{\{a_1,a_2\}} = \LsetFun{\{a_1,a_2,a_3\}}$,
whereas $\LsetFun{\{a_1,a_2\}} \circ \LsetFun{\{a_3\}} \neq
\LsetFun{\{a_1,a_2,a_3\}}$. This suggests that the order in which modules are
taken into account plays an important role in the definition of the composition
operator. Thus, we will focus on an \emph{ordered partition} $\pi = (X_1,
\ldots, X_m)$ of $A$, \ie a partition of $A$ provided with a strict total order
and represented by a sequence, called a \emph{modular organisation}, preserving
(\ref{eq:compo_part}). Furthermore, we would like to be able to ``fold''
contiguous modules in $\pi$ in order to deal with them as with a single
module\footnote{The folding of modules corresponds to the union of these
  modules.}.  As a consequence, we require a modular organisation to support
\emph{folding} and to be so the composition operator $\oslash$ is associative
according to the order in $\pi$.\smallskip
\begin{figure}[t!]
  \begin{center}
    \begin{minipage}{47mm} 
      \begin{tabular}{c}
        $\eta = \left\{
          \begin{array}{l}
            \eta_{a_1}(s) = \neg s_{a_1} \lor s_{a_2} \\
            \eta_{a_2}(s) = 1 \\
            \eta_{a_3}(s) = s_{a_2} \lor \neg s_{a_3} \\
          \end{array}
        \right. $\\[8mm]
        % \begin{tikzpicture}[scale=1, node distance=1.5cm]
        %   \SetVertexNormal[Shape=circle, LineWidth=0.5pt] 
        %   \Vertices[Math, x=0, y=0, dir=\EA]{a_1,a_2,a_3};
        %   \tikzset{EdgeStyle/.style = {post}}
        %   \draw (a_1) to [->, loop left, looseness=0, min distance=8mm, in=45, 
        %   out=135](a_1);
        %   \draw (a_3) to [->, loop left, looseness=0, min distance=8mm, in=45, 
        %   out=135](a_3);
        %   \draw (a_2) to [->] (a_1);
        %   \draw (a_2) to [->] (a_3);
        % \end{tikzpicture}
        \includegraphics[scale=1]{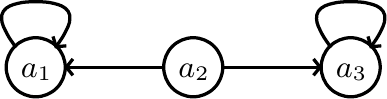}
      \end{tabular}
    \end{minipage}
    \hspace*{2mm}\begin{minipage}{70mm} 
      \vspace*{-1mm}\begin{tikzpicture}[xscale=1.6,yscale=.7]
        \scriptsize %% size of node labels
        \GraphInit[vstyle=Normal] %% type of graph
        \SetVertexNormal[Shape = rectangle, LineWidth=0.5pt]
        \node[rectangle,draw] at (0,.75) (000) {$\scriptstyle 000$};
        \node[rectangle,draw] at (3.75,.75) (001) {$\scriptstyle 001$};
        \node[rectangle,draw] at (0,3.5) (010) {$\scriptstyle 010$};
        \node[rectangle,draw] at (3.75,3.5) (011) {$\scriptstyle 011$};
        \node[rectangle,draw] at (1.1,1.4) (100) {$\scriptstyle 100$};
        \node[rectangle,draw] at (1.1,2.8) (110) {$\scriptstyle 110$};
        \node[rectangle,draw] at (2.65,1.4) (101) {$\scriptstyle 101$};
        \node[rectangle,draw,fill = black!20] at (2.65,2.8) (111) 
        {$\scriptstyle 111$};
        %% evolutions
        \tikzset{EdgeStyle/.style = {post}}
        \tikzset{LabelStyle/.style = {fill=white,sloped}}
        \tikzset{LabelStyle/.style = {color=white, text=black}}
        \Edge[label={$a_1$}](000)(100)
        \Edge[label={$a_2$}](000)(010)
        \Edge[label={$a_3$}](000)(001)
        \Edge[label={$a_1$}](001)(101)
        \Edge[label={$a_2$}](001)(011)
        \Edge[label={$a_3$}](001)(000)
        \Edge[label={$a_1$}](010)(110)
        \Edge[label={$a_3$}](010)(011)
        \Edge[label={$a_1$}](011)(111)
        \Edge[label={$a_1$}](100)(000)
        \Edge[label={$a_2$}](100)(110)
        \Edge[label={$a_3$}](100)(101)
        \Edge[label={$a_1$}](101)(001)
        \Edge[label={$a_2$}](101)(111)
        \Edge[label={$a_3$}](101)(100)
        \Edge[label={$a_3$}](110)(111)
      \end{tikzpicture}
    \end{minipage}
  \end{center}
  \vspace*{-4mm}\caption{An interaction network and its state graph.}
  \label{fig:examplec}
\end{figure}
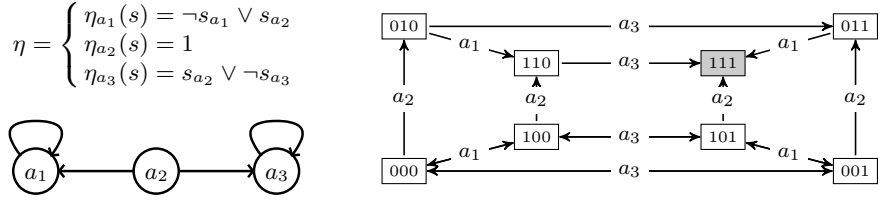

\noindent In order to form a modular organisation, the modules and their order
in $\pi$ should satisfy some conditions related to their dynamics.  Intuitively,
two disjoint sets of agents $X_i$ and $X_j$, ($i < j$), can be modules in $\pi$,
either if they do not regulate each other, or if $X_i$ regulates $X_j$. In both
cases, we can remark that the equilibria of $X_i$ should embed the asymptotic
evolution of $X_j$, which leads to encompass the equilibria of $X_j$ in the
equilibria of $X_i$. These conditions are expressed by the \emph{modularity
  relation} ($M$-relation).
\begin{definition}
  \label{def:m_relation}
  The \emph{$M$-relation}\footnote{Notice that the $M$-relation is reflexive
    but, in general, not transitive and that, when its definition holds for some
    $S'\subseteq S$, it holds for any subset of $S'$.} $\leadsto\
  \subseteq\ \partset{A} \times \partset{A}$ is defined as
  follows:%\vspace*{-1mm}
  \begin{equation*}
    X_i \leadsto X_j \eqdef \forall S' \subseteq S:
    (\LsetFun{X_i} \circ \Lset{X_i\cup X_j}{S'})\hspace*{-3pt}\evo_{X_j}\
    \subseteq\ (\LsetFun{X_i} \circ \Lset{X_i\cup X_j}{S'})\text{.}
    \end{equation*}
\end{definition}
\begin{proposition}
  \label{prop:leadsto}
  Let $S'$ be a subset of $S$. For all $X_i$, $X_j$ subsets of $A$, we have the
  following properties:
  \begin{enumerate}
  \item $X_i \leadsto X_j \iff \Lset{X_i \cup X_j}{S'} \subseteq (\LsetFun{X_i}
    \circ \orbit{X_i \cup X_j}{S'})$; \label{lorbit}
  \item $X_i \leadsto X_j \iff \Lset{X_i\cup X_j}{S'} = \LsetFun{X_i} \circ
    \Lset{X_i \cup X_j}{S'}$. \label{llset1}
  \end{enumerate}
\end{proposition}
\begin{proof}
  \begin{enumerate}
  \item ($\Rightarrow$) Let $S' \subseteq S$, $s \in \Lset{X_i\cup X_j}{S'}$
    and $s \notin \Lset{X_i}{\orbit{X_i \cup X_j}{S'}}$, and $s' \in
    \Lset{X_i}{\orbit{X_i \cup X_j}{\{s\}}}$. By definition of equilibrium, $s'
    \evo^*_{X_i\cup X_j} s$. Now, we have:
    \begin{itemize}
    \item $\forall s'' \in (s'\evo_{X_i}) : s'' \in \LsetFun{X_i} \circ
      \Lset{X_i\cup X_j}{\{s'\}}$, by definition of equilibria;
    \item $\forall s'' \in (s'\evo_{X_j}) : s'' \in \LsetFun{X_i} \circ
      \Lset{X_i\cup X_j}{\{s'\}}$, by definition of $\leadsto$.
  \end{itemize}
  As a consequence, $s' \nevo^*_{X_i\cup X_j} s$, which leads to a
  contradiction.\\[2mm]
  ($\Leftarrow$) Let $S' \subseteq S$, $s \in \Lset{X_i\cup X_j}{S'}$, and
  $\Lset{X_i\cup X_j}{S'} \subseteq \Lset{X_i}{\orbit{X_i \cup X_j}{S'}}$. Then,
  by hypothesis, we have:\\
  \centerline{ $\forall s'' \in (s\evo_{X_j}) : s''\in \Lset{X_i\cup X_j}{S'}
    \land s''\in \Lset{X_i}{\orbit{X_i \cup X_j}{S'}}\text{,}$ } which means
  that $s, s'' \in \LsetFun{X_i} \circ \Lset{X_i\cup X_j}{S'}$. Thus, $X_i
  \leadsto X_j$.\\[-2mm]
\item From Proposition~\ref{prop:leadsto}.\ref{lorbit} and since $\Lset{X_i
    \cup X_j}{S'} = \orbit{X_i \cup X_j}{\Lset{X_i\cup X_j}{S'}}$.\cqfd
\end{enumerate} 
\end{proof}
In this context, a modular organisation can be defined as:
\begin{definition}
  \label{def:modular-organisation}
  A \emph{modular organisation} is defined as an ordered partition $(X_1,
  \ldots, X_m)$ of $A$ such that, for all $1 \leq i \leq m$:
  $(\bigcup_{j=1}^{i-1} X_j) \leadsto X_i$.
\end{definition}
From Definition~\ref{def:modular-organisation},
Proposition~\ref{prop:modular-organisation} below states that being a modular
organisation is preserved by any folding of its contiguous parts.
\begin{proposition}
  \label{prop:modular-organisation}
  Let $\pi = (X_1, \ldots, X_m)$ be a modular organisation. For all $1 \leq i
  \leq j \leq m$, $(X_1, \ldots, X_{i-1}, \bigcup_{k=i}^j X_k, X_{j+1}, \ldots,
  X_m)$ is a modular organisation.
\end{proposition}
\begin{proof}
  Let $\pi = (X_1, \ldots, X_{i-1}, X_i, X_{i+1}, \ldots, X_m)$ be a modular
  organisation. Let $X = \bigcup_{k=1}^{i-1} X_k$.  We want to show that $(X_1,
  \ldots, X_{i-1}, X_i \cup X_{i+1}, \ldots, X_m)$ is a modular organisation. By
  definition~\ref{def:modular-organisation}, we have:
  \begin{equation}
    \label{eq01}
    (X \cup X_i) \leadsto X_{i+1}
  \end{equation}
  and:
    \begin{equation}
    \label{eq02}
    X \leadsto X_i\text{.}
  \end{equation}
  We want to show that: $\text{(\ref{eq01})} \land \text{(\ref{eq02})} \implies
  X \leadsto (X_i \cup X_{i+1})$. First, by
  Proposition~\ref{prop:leadsto}.\ref{llset1}, we can write:
  \begin{equation}
    \label{eq03}
    \LsetFun{X \cup X_i} \circ \LsetFun{X \cup X_i \cup X_{i+1}} =
    \LsetFun{X \cup X_i \cup X_{i+1}} \text{ by (\ref{eq01}),}
  \end{equation}
  \begin{equation}
    \label{eq04}
    \LsetFun{X} \circ \LsetFun{X \cup X_i} = \LsetFun{X \cup X_i} 
    \text{ by (\ref{eq02}).}
  \end{equation}
  Thus:
  \begin{equation*}
    \begin{split}
      \LsetFun{X} \circ \LsetFun{X \cup X_i \cup X_{i+1}} & = \LsetFun{X} \circ
      (\LsetFun{X \cup X_i} \circ
      \LsetFun{X \cup X_i \cup X_{i+1}}) \text{ by (\ref{eq03})}\\
      & = (\LsetFun{X} \circ \LsetFun{X \cup X_i}) \circ
      \LsetFun{X \cup X_i \cup X_{i+1}}\\
      & = \LsetFun{X \cup X_i} \circ \LsetFun{X \cup X_i \cup X_{i+1}} 
      \text{ by (\ref{eq04})}\\
      & = \LsetFun{X \cup X_i \cup X_{i+1}} \text{ by (\ref{eq03}),}
    \end{split}
  \end{equation*}
  which is the expected result. From
  Proposition~\ref{prop:leadsto}.\ref{llset1}, we can deduce that: $X \leadsto
  (X_i \cup X_{i+1})$. Iteratively, we show that $X \leadsto \bigcup_{k=i}^j
  X_k$. As a result, $(X_1, \ldots, X_{i-1}, \bigcup_{k=i}^j X_k, X_{j+1}, \ldots,
  X_m)$ is a modular organisation.\cqfd
\end{proof}
In the literature~\cite{Gagneur2004,Milo2002,Alon2003}, modules are frequently
related assimilated to SCCs of interaction networks. Although these works focus
on structural arguments only, it turns out that they are compatible with
Definition~\ref{def:modular-organisation}. Indeed, any topological ordering of
SCCs is actually a modular organisation. For instance, $(\{a_1\}, \{a_2,a_3\},
\{a_4\})$ is a modular organisation of the interaction network presented in
Figure~\ref{fig:examplea}. Similarly, any other structural or dynamical property
could be helpful in the research of modular organisations. In what follows, we
present an approach addressing formally this aspect. As a result, we show that,
in particular, the structural decomposition in SCCs makes sense and may be
improved by a deeper analysis leading to the decomposition of SCCs in elementary
modules (see Section~\ref{sec:separation}).

\subsection{Regulation and modularity relation}
\label{sec:regulation-Mrelation}

In this section, we focus on the relation between the notions of regulation and
$M$-relation. More precisely, and unexpectedly, we will see in
Lemma~\ref{lem:nonregul-is-m} that this relation refers more directly to the
absence of regulation than to the regulation itself.
\begin{lemma}
  \label{lem:nonregul-is-m}
  For any $X_i$, $X_j$ subsets of $A$: $\neg(X_j \reg X_i) \implies X_i \leadsto
  X_j$.
\end{lemma}
\begin{proof}
  By Definition~\ref{def:reg}, for any $X_i, X_j \subseteq A$ and for any $s, s'
  \in S$, we have $\neg (X_j \reg X_i) \land (s \eqr {A \diff X_j } s') \implies
  (s \evo_{X_i}) \eqr {X_i} (s' \evo_{X_i})$. This property is obviously
  preserved at equilibria. Indeed, for any $s, s' \in S$, we have $\neg (X_j
  \reg X_i) \land (s \eqr {A \diff X_j} s') \implies \Lset{X_i}{s} \eqr {X_i}
  \Lset{X_i}{s'}$. Thus, the restrictions $\Lset{X_i}{s}$ and $\Lset{X_i}{s'}$
  to $X_i$ are identical. Then, the evolution by $X_j$ from the equilibria of
  $X_i$ remains in the equilibria of $X_i$. Hence, we get that $\neg(X_j \reg
  X_i) \implies X_i \leadsto X_j.$\cqfd
\end{proof}
Also, as a remarkable fact, according to
Definition~\ref{def:modular-organisation}, Theorem~\ref{the:scc-folfing}
provides a connection between structural properties of a regulation graph and
the corresponding modular organisations (possibly reduced to a single module).
\begin{theorem}
  \label{the:scc-folfing}
  Any topological ordering of the SCC quotient graph of a regulation graph is a
  modular organisation.
\end{theorem}
\begin{proof}
  Observe that, in the SCC quotient graph $G$ of a regulation graph, $X_i \reg
  X_j$ always implies that $\neg (X_j \reg X_i)$, because of the acyclicity of
  $G$. Thus, folding contiguous modules with respect to any topological ordering
  preserves the absence of regulation. As a consequence, if $(X_1, \ldots, X_m)$
  is a topological ordering of $G$, for all $i, j \in \Nset$ such that $1 \leq i
  \leq j \leq m$, we have $\neg(X_j \reg X_i)$, and by
  Lemma~\ref{lem:nonregul-is-m}, $X_i \leadsto X_j$.\cqfd
\end{proof}

\subsection{Composition operator}
\label{sec:composition-operator}

\begin{figure}[t!]
  \centerline{
    \begin{tabular}{ccc} 
      \includegraphics[scale=0.79]{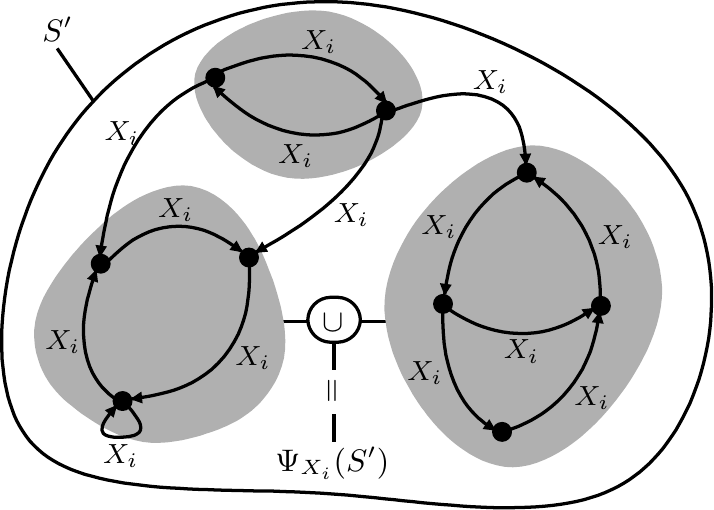} & 
      ~~~~~ & 
      \includegraphics[scale=0.79]{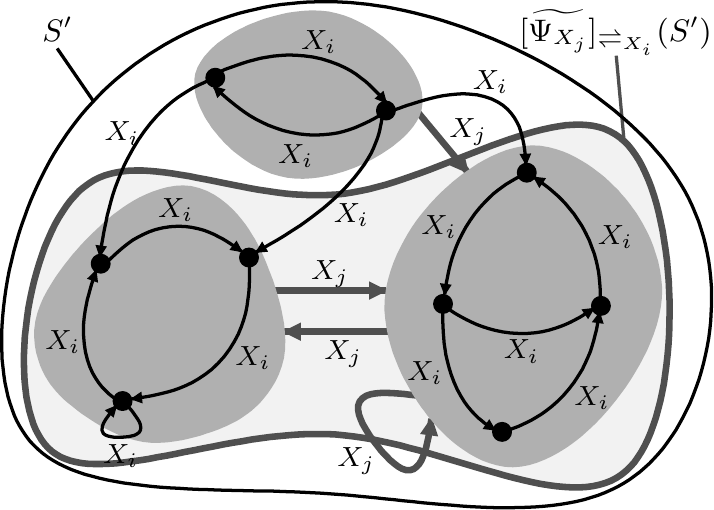}\\
      $(a)$ & & $(b)$
    \end{tabular}
    }
  \caption{Successive steps leading to the definition of the composition
    operator $\oslash$.}
  \label{fig:composition_operator}
\end{figure}

In this section, we present the successive steps leading to the definition of
the composition operator $\oslash$. From (\ref{eq:compo_part}), $\oslash$ is a
binary operator that applies on the equilibria of parts $X_i$ and $X_j$ of
$\pi$, with $i<j$. Thus, its definition is based on the attractors of
$\evo_{X_i}$ which correspond to terminal nodes of the SCC quotient graph of
$\evo_{X_i}$, namely $\mathcal{G}/{}_{\rightleftharpoons_{X_i}}$, where
$\rightleftharpoons_{X_i}$ is the equivalence relation on states defined as $s
\rightleftharpoons_{X_i} s' \eqdef (s \evo^*_{X_i} s') \land (s' \evo^*_{X_i}
s)$. For any $S' \subseteq S$, an attractor of $\evo_{X_i}$ complies to:
$[s]_{\rightleftharpoons_{X_i}} \subseteq \LsetFun{X_i}(S')$ (see
Figure~\ref{fig:composition_operator}.a). Moreover, for all $S'
\subseteq S$, we denote by:\\[-7mm]
\begin{itemize}
\item $[S']_{\rightleftharpoons_{X_i}} = \{[s]_{\rightleftharpoons_{X_i}} | s
  \in S'\}$ the sets of equivalence classes of $\rightleftharpoons_{X_i}$ in
  $S'$;
\item $[s]_{\rightleftharpoons_{X_i}} \; [\evo_{X_j}]_{\rightleftharpoons_{X_i}}
  \; [s']_{\rightleftharpoons_{X_i}} \eqdef \exists s \in
  [s]_{\rightleftharpoons_{X_i}}, \exists s' \in
  [s']_{\rightleftharpoons_{X_i}}: s \evo_{X_j} s'$ the evolution by agents of
  $X_j$ on these equivalence classes.\\[-7mm]
\end{itemize}
We define an operator $[\widetilde{\LsetFun{X_j}}]_{\rightleftharpoons_{X_i}}$
that computes the set of equilibria of $[\evo_{X_j}]_{\rightleftharpoons_{X_i}}$
in $\mathcal{G}/{}_{\rightleftharpoons_{X_i}}$ (see
Figure~\ref{fig:composition_operator}.b) as follows:
\begin{multline*}
  [\widetilde{\LsetFun{X_j}}]_{\rightleftharpoons_{X_i}}(S')\ \eqdef\ \{ e \in
  [S']_{\rightleftharpoons_{X_i}} \mid ((e
  [\evo_{X_j}]_{\rightleftharpoons_{X_i}}^*) \subseteq
  [S']_{\rightleftharpoons_{X_i}})\ \land\\ (\forall e' \in
  [S]_{\rightleftharpoons_{X_i}}: e [\evo_{X_j}]_{\rightleftharpoons_{X_i}}^* e'
  \implies e' [\evo_{X_j}]_{\rightleftharpoons_{X_i}}^* e)\}.
\end{multline*}
Operator $\oslash$ is thus defined as:
\begin{equation}
  \label{eq:op}
  \LsetFun{X_i} \oslash \LsetFun{X_j} = \textsf{Flatten} \circ 
  [\widetilde{\LsetFun{X_j}}]_{\rightleftharpoons_{X_i}} \circ \LsetFun{X_i}\text{,}
\end{equation}
where, for any set $E\subseteq \partset{S}$, $\textsf{Flatten}(E) = \bigcup_{e
  \in E} e$. Thus, if applied to a set of attractors, \textsf{Flatten} gives the
underlying set of equilibria\footnote{For example,
  $\textsf{Flatten}(\{\{00,01\},\{10\}\}) = \{00,01,10\}$.}. As a result, it can
easily be seen that $\oslash$ does compute the set of states belonging to the
attractors of $X_i$ which are also the equilibria of $\evo_{X_j}$.\smallskip

\noindent Lemma~\ref{lem:ccmod} below shows that the global equilibria of $X_i
\cup X_j$ are obtained using the composition $\LsetFun{X_i} \oslash
\LsetFun{X_j}$.
\begin{lemma} 
  \label{lem:ccmod}
  For all $X_i$, $X_j$ disjoint subsets of $A$:%\vspace*{-2mm}
  \begin{equation*}
    X_i \leadsto X_j \implies \forall S' \subseteq S:
    \LsetFun{X_i \cup X_j} (S') = (\LsetFun{X_i} \oslash \LsetFun{X_j}) \circ
    \orbitFun{X_i \cup X_j}(S')\text{.}
  \end{equation*}
\end{lemma}
\begin{proof}
  Let $X_i, X_j \subseteq A$ such that $X_i \leadsto X_j$ and $S'$ a subset of
  $S$.\\
  \textbf{($\subseteq$)} First, let us show $\LsetFun{X_i \cup X_j} (S')
  \subseteq (\LsetFun{X_i} \oslash \LsetFun{X_j}) \circ \orbitFun{X_i \cup
    X_j}(S')$. From Proposition~\ref{prop:leadsto}.\ref{llset1}, we know that
  $\LsetFun{X_i} \circ \Lset{X_i \cup X_j}{S'} = \Lset{X_i \cup X_j}{S'}$. Thus,
  $\forall s \in \LsetFun{X_i \cup X_j}(S')$, we have
  $[s]_{\rightleftharpoons_{X_i}} \subseteq \LsetFun{X_i \cup
    X_j}(S')$. Similarly, an evolution by $X_i \cup X_j$ from an attractor of
  $X_i$ remain in the same attractor except potentially with evolutions by
  $X_j$. We have then, for all $s,s' \in \LsetFun{Xi \cup X_j}(S')$:
  \begin{multline*}
    [s]_{\rightleftharpoons_{X_i}} [\evo_{X_i \cup
      X_j}]_{\rightleftharpoons_{X_i}}^* [s']_{\rightleftharpoons_{X_i}} \iff \\
    ([s]_{\rightleftharpoons_{X_i}} [\evo_{X_j}]_{\rightleftharpoons_{X_i}}^*
    [s']_{\rightleftharpoons_{X_i}})\ \lor\ ([s]_{\rightleftharpoons_{X_i}} =
    [s']_{\rightleftharpoons_{X_i}})\text{.}
  \end{multline*}
  \noindent Now, since both $s$ and $s'$ belong to attractors of $X_i \cup X_j$
  (by hypothesis), if there exists an evolution by $X_j$ from
  $[s]_{\rightleftharpoons_{X_i}}$ to $[s']_{\rightleftharpoons_{X_i}}$, there
  exists obviously another path labelled by $X_j$ from
  $[s']_{\rightleftharpoons_{X_i}}$ to $[s]_{\rightleftharpoons_{X_i}}$. Hence,
  for all $s,s' \in \LsetFun{X_i \cup X_j}(S')$, we have:
  \begin{multline*}
    ([s]_{\rightleftharpoons_{X_i}} [\evo_{X_j}]_{\rightleftharpoons_{X_i}}^*
    [s']_{\rightleftharpoons_{X_i}}) = ([s]_{\rightleftharpoons_{X_i}} [\evo_{X_i
      \cup X_j}]_{\rightleftharpoons_{X_i}}^* [s']_{\rightleftharpoons_{X_i}})\ \implies\\
    ([s']_{\rightleftharpoons_{X_i}} [\evo_{X_i \cup X_j}]_{\rightleftharpoons_{X_i}}^* 
    [s]_{\rightleftharpoons_{X_i}}) = ([s']_{\rightleftharpoons_{X_i}} 
    [\evo_{X_j}]_{\rightleftharpoons_{X_i}}^* [s]_{\rightleftharpoons_{X_i}})\text{.}
  \end{multline*}
  \noindent As a result, we have $[s]_{\rightleftharpoons_{X_i}} \in
  [\widetilde{\LsetFun{X_j}}]_{\rightleftharpoons_{X_i}} ([\Lset{X_i \cup
    X_j}{S'}]_{\rightleftharpoons_{X_i}})$, for all $s \in \Lset{X_i \cup
    X_j}{S'}$. Moreover, since from Proposition~\ref{prop:eq}, operator
  $[\widetilde{\LsetFun{X_j}}]_{\rightleftharpoons_{X_i}}(S')$ is
  monotonous 
  % La référence donnée pour la monotonie ne correspond pas à l'opération
  % mais...
  and since $\Lset{X_i \cup X_j}{S'} \subseteq \orbit{X_i \cup X_j}{S'}$, for
  all $s \in \Lset{X \cup Y}{S'}$, we can write that
  $[s]_{\rightleftharpoons_{X_i}} \in
  [\widetilde{\LsetFun{X_j}}]_{\rightleftharpoons_{X_i}} (\orbit{X_i \cup
    X_j}{S'})$. Now, since $s \in [s]_{\rightleftharpoons_{X_i}}$, $\forall s
  \in \Lset{X_i \cup X_j}{S'}$, we have $s\in \textsf{Flatten} \circ
  [\widetilde{\LsetFun{X_j}}]_{\rightleftharpoons_{X_i}}(\orbit{X_i \cup
    X_j}{S'})$.
  \noindent From (\ref{eq:op}) and Proposition~\ref{prop:leadsto}.\ref{llset1},
  we can write:
  \begin{equation*}
    \begin{split}
      \LsetFun{X_i} \oslash \LsetFun{X_j} \circ \orbitFun{X_i \cup X_j}(S') & =
      \textsf{Flatten} \circ
      [\widetilde{\LsetFun{X_j}}]_{\rightleftharpoons_{X_i}} \circ \LsetFun{X_i}
      \circ \orbitFun{X_i \cup X_j}(S')\\
      & = \textsf{Flatten} \circ
      [\widetilde{\LsetFun{X_j}}]_{\rightleftharpoons_{X_i}} \circ \orbitFun{X_i
        \cup X_j}(S')\text{.}
    \end{split}
  \end{equation*}
  Hence, for all $s \in \Lset{X \cup Y}{S'}$, we have: $s \in \LsetFun{X}
  \oslash \LsetFun{Y} \circ \orbitFun{X \cup Y}(S')$, which corresponds to the
  following inclusion:
  \begin{equation*}
    \LsetFun{X \cup Y}(S') \subseteq \LsetFun{X} \oslash \LsetFun{Y} \circ 
    \orbitFun{X \cup Y}(S')\text{.}
  \end{equation*}
  \noindent \textbf{($\supseteq$)} Now, let us show $(\LsetFun{X_i} \oslash
  \LsetFun{X_j}) \circ \orbit{X_i \cup X_j}{S'} \subseteq \Lset{X_i \cup
    X_j}{S'}$. To do so, let us consider a state $s \in (\LsetFun{X_i} \oslash
  \LsetFun{X_j}) \circ \orbitFun{X_i \cup X_j}(S')$. From (\ref{eq:op}) and
  Proposition~\ref{prop:leadsto}.\ref{llset1}, we have
  $[s]_{\rightleftharpoons_{X_i}} \in [\LsetFun{X_j}]_{\rightleftharpoons_{X_i}}
  \circ [\LsetFun{X_i} \circ \orbitFun{X_i \cup
    X_j}(S')]_{\rightleftharpoons_{X_i}}$. Now, consider
  $[[s]_{\rightleftharpoons_{X_i}}]_{\rightleftharpoons_{X_j}}$. By definition
  of attractors, for all $s_1, s_2 \in \textsf{Flatten} \circ
  \textsf{Flatten}([[s]_{\rightleftharpoons_{X_i}}]_{\rightleftharpoons_{X_j}})$,
  we have $s_1 \evo^*_{X_i \cup X_j} s_2$. This means that
  $[[s]_{\rightleftharpoons_{X_i}}]_{\rightleftharpoons_{X_j}} =
  [s]_{\rightleftharpoons_{X_i\cup X_j}}$ and, as a consequence, that $s\in
  \LsetFun{X \cup Y} (S')$. As a result, the inclusion $(\LsetFun{X} \oslash
  \LsetFun{Y}) \circ \orbitFun{X \cup Y}(S') \subseteq \LsetFun{X \cup Y} (S')$
  holds.\cqfd
\end{proof}
As a main result, Theorem~\ref{the:gamma} shows how $\oslash$ can be used on a
modular organisation $\pi$ to obtain a modular computation of global equilibria
of an interaction network.
\begin{theorem}
  \label{the:gamma}
  Let $A' = \bigcup_{i=1}^m X_i\subseteq A$ be a set of agents such that
  $(X_1, \ldots,$  $X_m)$ is a modular organisation, we have:
  \begin{equation*}
    \LsetFun{A'} = (\LsetFun{X_1} \oslash \ldots \oslash \LsetFun{X_m})
    \circ \orbitFun{A'}\text{.}
  \end{equation*}
\end{theorem}  
\begin{proof}
  This proof is made directly, using Definition~\ref{def:modular-organisation}
  and Lemma~\ref{lem:ccmod}. Since $\pi = (X_1, \ldots, X_m)$ is a modular
  organisation, it is folding preserving and $\bigcup_{i=1}^{m-1} X_i \leadsto
  X_{m}$. Then, using (\ref{eq:op}) and Lemma~\ref{lem:ccmod}, we have:
  \begin{gather*}
    \hspace*{-65mm}\LsetFun{A'} = (\LsetFun{\bigcup_{i=1}^{m-1} X_i} \oslash
    \LsetFun{X_m}) \circ \orbitFun{A'}\\
    \hspace*{6mm}\begin{split}
      & = \textsf{Flatten} \circ
      [\widetilde{\LsetFun{X_m}}]_{\rightleftharpoons_{\bigcup_{i=1}^{m-1}X_i}}
      \circ
      \LsetFun{\bigcup_{i=1}^{m-1} X_i} \circ \orbitFun{A'}\\
      & = \textsf{Flatten} \circ
      [\widetilde{\LsetFun{X_m}}]_{\rightleftharpoons_{\bigcup_{i=1}^{m-1} X_i}}
      \circ (\LsetFun{\bigcup_{i=1}^{m-2} X_i} \oslash \LsetFun{X_{m-1}}) \circ
      \orbitFun{\bigcup_{i=1}^{m-1} X_i} \circ \orbitFun{A'}\\
      & = \textsf{Flatten} \circ
      [\widetilde{\LsetFun{X_m}}]_{\rightleftharpoons_{\bigcup_{i=1}^{m-1} X_i}}
      \circ (\LsetFun{\bigcup_{i=1}^{m-2} X_i} \oslash \LsetFun{X_{m-1}}) \circ
      \orbitFun{A'}\\
      & = \textsf{Flatten} \circ
      [\widetilde{\LsetFun{X_m}}]_{\rightleftharpoons_{\bigcup_{i=1}^{m-1} X_i}}
      \circ
      [\widetilde{\LsetFun{X_{m-1}}}]_{\rightleftharpoons_{\bigcup_{i=1}^{m-2}
          X_i}} \circ \LsetFun{\bigcup_{i=1}^{m-2} X_i} \circ \orbitFun{A'}\\
      & = \ldots\\
      & = \textsf{Flatten} \circ
      [\widetilde{\LsetFun{X_m}}]_{\rightleftharpoons_{\bigcup_{i=1}^{m-1} X_i}}
      \circ \ldots \circ [\widetilde{\LsetFun{X_2}}]_{\rightleftharpoons_{X_1}}
      \circ \LsetFun{X_1} \circ \orbitFun{A'}\text{.}
    \end{split}
  \end{gather*}
  As a result, we obtain:
  \begin{equation*}
    \Lset{A'}{S'} = (\LsetFun{X_1} \oslash \ldots \oslash \LsetFun{X_m})
    \circ \orbit{A'}{S'}\text{,}
  \end{equation*}
  which is the expected result.\cqfd
\end{proof}

\section{Elementary modular organisation}
\label{sec:separation}

Informally, a module is \emph{elementary} if it is not separable, \ie if the
equilibria of each of its agents depend entirely on the equilibria of all the
others. For instance, consider negative circuits that lead to asymptotic
sustained oscillations~\cite{Thomas1981}. In such regulation patterns,
the equilibria of an agent cannot be encompassed into that of the others
because, in order to reach its own equilibria, each agent evolves from the
equilibria of all the others.\smallskip

\noindent In this context, a modular organisation provided by some topological
ordering of the SCC quotient graph (see Theorem~\ref{the:scc-folfing}) does
not always provide an elementary decomposition. So, it may be not satisfactory
since modules often support specific biological functions, revealing
\emph{elementary} modules appears important, not to say essential, to highlight
basic biological functional mechanisms.\smallskip
\begin{figure}[t!]
  \begin{center}
    \begin{minipage}{49.55mm}
      \begin{tabular}{c}
        $\eta = \left\{
          \begin{array}{l}
            \eta_{a_1}(s) = (s_{a_1} \land s_{a_2}) \lor \neg s_{a_3}\\
            \eta_{a_2}(s) = s_{a_1} \land s_{a_2} \land s_{a_3}\\
            \eta_{a_3}(s) = s_{a_1}\\
          \end{array}
        \right.$\\[8mm]
        % \hspace*{-4mm}\begin{tikzpicture}[scale=1, node distance=1.5cm]
        %   \SetVertexNormal[Shape=circle, LineWidth=0.5pt] 
        %   \Vertices[Math, x=0, y=0, dir=\EA]{a_2,a_1,a_3};
        %   \tikzset{EdgeStyle/.style = {post}}
        %   \draw (a_1) to [->, loop left, looseness=0, min distance=8mm, in=45, 
        %   out=135](a_1);
        %   \draw (a_2) to [->, loop left, looseness=0, min distance=8mm, in=45, 
        %   out=135](a_2);
        %   \draw (a_1) to [->, bend left] (a_2);
        %   \draw (a_2) to [->, bend left] (a_1);
        %   \draw (a_1) to [->, bend left] (a_3);
        %   \draw (a_3) to [->, bend left] (a_1);
        %   \draw (a_3) to [->, bend left=50] (a_2);
        % \end{tikzpicture}
        \includegraphics[scale=1]{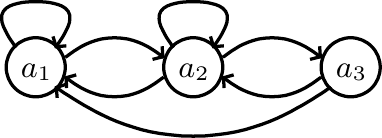}
      \end{tabular}
    \end{minipage}
    \hspace*{6mm}\begin{minipage}{55mm}
      \vspace*{0mm}\begin{tikzpicture}[xscale=1.3,yscale=.7]
        \scriptsize %% size of node labels
        \GraphInit[vstyle=Normal] %% type of graph
        \SetVertexNormal[Shape = rectangle, LineWidth=0.5pt]
        \node[rectangle, draw, color=white, fill=black!65] at (0,.75) (000) 
        {$\scriptstyle 000$};
        \node[rectangle, draw, color=white, fill=black!65] at (3.75,.75) (001) 
        {$\scriptstyle 001$};
        \node[rectangle, draw] at (0,3.5) (010) {$\scriptstyle 010$};
        \node[rectangle, draw] at (3.75,3.5) (011) {$\scriptstyle 011$};
        \node[rectangle, draw, color=white, fill=black!65] at (1.1,1.4) (100) 
        {$\scriptstyle 100$};
        \node[rectangle, draw] at (1.1,2.8) (110) {$\scriptstyle 110$};
        \node[rectangle, draw, color=white, fill=black!65] at (2.65,1.4) (101) 
        {$\scriptstyle 101$};
        \node[rectangle, draw, fill = black!20] at (2.65,2.8) (111) 
        {$\scriptstyle 111$};
        %% evolutions
        \tikzset{EdgeStyle/.style = {post}}
        \tikzset{LabelStyle/.style = {fill=white,sloped}}
        \tikzset{LabelStyle/.style = {color=white, text=black}}
        \Edge[label={$a_1$}](000)(100)
        \Edge[label={$a_1$}](010)(110)
        \Edge[label={$a_1$}](101)(001)
        \Edge[label={$a_2$}](010)(000)
        \Edge[label={$a_2$}](011)(001)
        \Edge[label={$a_2$}](110)(100)
        \Edge[label={$a_3$}](001)(000)
        \Edge[label={$a_3$}](011)(010)
        \Edge[label={$a_3$}](100)(101)
        \Edge[label={$a_3$}](110)(111)
      \end{tikzpicture}
    \end{minipage}
  \end{center}
  \vspace*{-5mm}\caption{A separable interaction network with modular
    organisation $(\{b\}, \{a, c\})$.}
  \label{fig:exampled}
\end{figure}
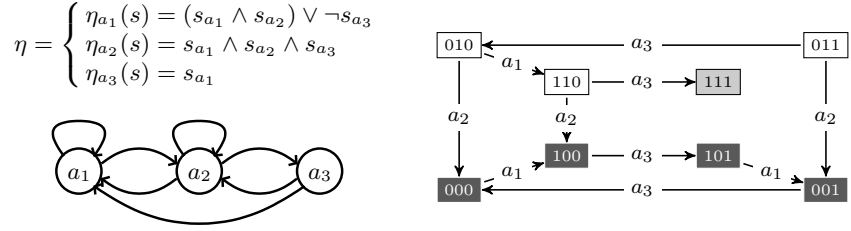

\noindent Figure~\ref{fig:exampled} depicts an interaction network $\eta$
composed of three agents $a_1$, $a_2$ and $a_3$ with the associated strongly
connected regulation graph. Its underlying state graph shows that the global
dynamics of $\eta$ leads to two attractors, stable state $\{111\}$ and limit set
$\{000, 100, 101, 001\}$. It is easy to see that $\{a_2\} \leadsto \{a_1,a_3\}$
and that this $M$-relation is (obviously) preserved by folding, because there
are only two modules. Hence, ordered partition $(\{a_2\}, \{a_1, a_3\})$ is a
modular organisation of $\eta$.\smallskip

\noindent However, separability is not possible in general as illustrated in
Figure~\ref{fig:example-not-sep}. Indeed, starting from modular organisation
$\pi = (\{a_1\},\{a_2, a_3\})$ obtained from the SCCs, the separation of $\{a_2,
a_3\}$ should lead to the ordered partition $\pi'=(\{a_1\},\{a_2\},\{a_3\})$.
The condition for $\pi'$ to be a modular organisation is that $\{a_1, a_2\}
\leadsto \{a_3\}$, \ie the evolution by $a_3$ from the equilibria of $\{a_1,
a_2\}$ have to be included in the equilibria of $\{a_1, a_2\}$. We can observe
that the attractors for $\{a_1, a_2\}$ are $\{010,011\}$ and $\{100,101\}$,
while the evolution by $a_3$ from either $011$ or $101$ leaves the attractors of
$\{a_1, a_2\}$, which means that $\pi'$ is not a modular organisation. As a
consequence, $\{a_2, a_3\}$ cannot be separated.\smallskip

\noindent Hence, the separation condition of a module $X_i$ in $\pi$ is not
local to this module but depends on the module ``context'', that is the global
equilibria (\ie $\LsetFun{\bigcup_{k=1}^{i-1} X_k}$) of the modules that precede
it in $\pi$. Deciding the separability of $X_i$ into $X_i^1$ and $X_i^2$ implies
checking two conditions: $\bigcup_{k=1}^{i-1} X_k \leadsto X_i^1$ and
$\bigcup_{k=1}^{i-1} X_k \cup X_i^1 \leadsto X_i^2$. Of course, the complexity
of the underlying computation is exponential in the size of $\pi$ and depends
also on the position of $X_i$ in $\pi$. Nevertheless, brute-force computation
may be used in practice for small interaction networks (of about $15$ agents). A
more efficient method allowing to go beyond this limitation is, for the moment,
an open question.
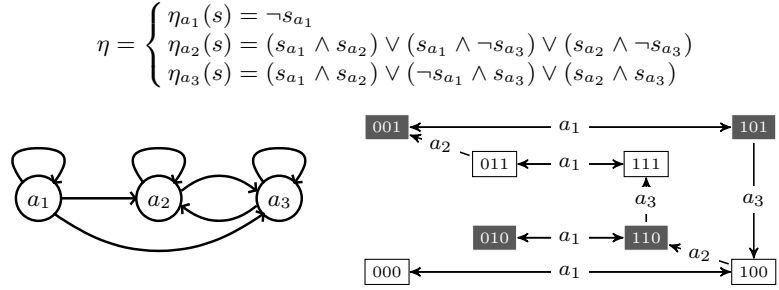
\begin{figure}[t!]
  \begin{center}
    \begin{tabular}{c}
      $\eta = \left\{
        \begin{array}{l}
          \eta_{a_1}(s) = \neg s_{a_1}\\
          \eta_{a_2}(s) = (s_{a_1} \land s_{a_2}) \lor (s_{a_1} \land \neg s_{a_3})
          \lor (s_{a_2} \land \neg s_{a_3})\\
          \eta_{a_3}(s) = (s_{a_1} \land s_{a_2}) \lor (\neg s_{a_1} \land s_{a_3})
          \lor (s_{a_2} \land s_{a_3})\\
        \end{array}
      \right.$
    \end{tabular}\vspace*{3mm}

    \begin{minipage}{108mm}
      % \hspace*{-1pt}
      % \begin{tikzpicture}[scale=1, node distance=1.5cm]
      %   \SetVertexNormal[Shape=circle, LineWidth=0.5pt] 
      %   \Vertices[Math, x=0, y=0, dir=\EA]{a_1,a_2,a_3};
      %   \tikzset{EdgeStyle/.style = {post}}
      %   \draw (a_1) to [->, loop left, looseness=0, min distance=8mm, in=45, 
      %   out=135](a_1);
      %   \draw (a_2) to [->, loop left, looseness=0, min distance=8mm, in=45, 
      %   out=135](a_2);
      %   \draw (a_3) to [->, loop left, looseness=0, min distance=8mm, in=45, 
      %   out=135](a_3);
      %   \draw (a_1) to [->] (a_2);
      %   \draw (a_1) to [->, bend right=50] (a_3);
      %   \draw (a_2) to [->, bend left] (a_3);
      %   \draw (a_3) to [->, bend left] (a_2);
      % \end{tikzpicture}
      \centerline{
        \begin{tabular}{m{4cm}m{0.5cm}m{5.6cm}}
          \includegraphics[scale=1]{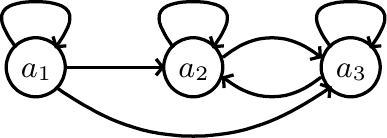} & ~ &
          \begin{tikzpicture}[xscale=1.3,yscale=.7]
            \scriptsize %% size of node labels
            \GraphInit[vstyle=Normal] %% type of graph
            \SetVertexNormal[Shape = rectangle, LineWidth=0.5pt]
            \node[rectangle, draw] at (0,.75) (000) {$\scriptstyle 000$};
            \node[rectangle, draw] at (3.75,.75) (100) {$\scriptstyle 100$};
            \node[rectangle, draw, color=white, fill=black!65] at (0,3.5) (001) 
            {$\scriptstyle 001$};
            \node[rectangle, draw, color=white, fill=black!65] at (3.75,3.5) 
            (101)
            {$\scriptstyle 101$};
            \node[rectangle, draw, color=white, fill=black!65] at (1.1,1.4) 
            (010) 
            {$\scriptstyle 010$};
            \node[rectangle, draw, color=white, fill=black!65] at (2.65,1.4)
            (110)
            {$\scriptstyle 110$};
            \node[rectangle, draw] at (1.1,2.8) (011) {$\scriptstyle 011$};
            \node[rectangle, draw] at (2.65,2.8) (111) 
            {$\scriptstyle 111$};
            %% evolutions
            \tikzset{EdgeStyle/.style = {post}}
            \tikzset{LabelStyle/.style = {fill=white,sloped}}
            \tikzset{LabelStyle/.style = {color=white, text=black}}
            \Edge[label={$a_1$}](000)(100)
            \Edge[label={$a_1$}](100)(000)
            \Edge[label={$a_1$}](001)(101)
            \Edge[label={$a_1$}](101)(001)        
            \Edge[label={$a_1$}](111)(011)
            \Edge[label={$a_1$}](011)(111)
            \Edge[label={$a_1$}](110)(010)
            \Edge[label={$a_1$}](010)(110)
            \Edge[label={$a_2$}](011)(001)
            \Edge[label={$a_2$}](100)(110)
            \Edge[label={$a_3$}](110)(111)
            \Edge[label={$a_3$}](101)(100)
          \end{tikzpicture}
        \end{tabular}
      }
    \end{minipage}
  \end{center}
  \vspace*{-2mm}\caption{Example of non-separability of $\{a_2, a_3\}$ in $\pi =
    (\{a_1\}, \{a_2, a_3\})$.}
  \label{fig:example-not-sep}
\end{figure}

\section{Conclusion}
\label{sec:discussion}

We developed a formal framework for the analysis of the modularity in
interaction networks assuming asymptotic dynamics of modules and enabling their
composition. We exhibited modularity conditions governing the composition of
modules and an efficient computation method such that the global equilibria of
interaction networks are obtained from the local ones. Moreover, we confirmed
that usual assumptions identifying modules to SCCs have a strong motivation
coming from theory. Of course, the next step should be finding ingenious
algorithms for searching elementary modular organisations. Then, since this work
provides a rigorous setting for studying other questions around modularity, it
would be of interest to exhibit elementary biological components, which would be
relevant, in particular, for synthetic biology~\cite{Purnick2009}. Also,
questions related to robustness and evolution could be tackled thanks to the
modular knowledge of interaction networks.

\section*{Acknowlegments}

This work was partially supported by the \emph{Agence nationale de la recherche}
and the \emph{R{\'e}seau national des syst{\`e}mes complexes} through the
respective projects Synbiotic (\textsc{anr} 2010 \textsc{blan} 0307\,01) and
M{\'e}t{\'e}ding (\textsc{rnsc} \textsc{ai}10/11-\textsc{l}03908).

\bibliographystyle{splncs} 
\bibliography{modularity}

\begin{thebibliography}{10}

\bibitem{Monod1970}
Monod, J.:
\newblock {Chance and necessity: an essay on the natural philosophy of modern
  biology}.
\newblock Knopf (1970)

\bibitem{Peter2009}
Peter, I.S., Davidson, E.H.:
\newblock {Modularity and design principles in the sea urchin embryo gene
  regulatory network}.
\newblock FEBS letters \textbf{583} (2009)  3948--3958

\bibitem{McDougall2011}
McDougall, C., Degnan, B.M.:
\newblock {Modularity of gene-regulatory networks revealed in sea-star
  development}.
\newblock BMC Biology \textbf{9} (2011)  \hspace*{--1pt}6

\bibitem{Qi2006}
Qi, Y., Ge, H.:
\newblock {Modularity and dynamics of cellular networks}.
\newblock PLoS Computational Biology \textbf{2} (2006)  e174

\bibitem{Gagneur2004}
Gagneur, J., Krause, R., Bouwmeester, T.,  et~al.:
\newblock {Modular decomposition of protein-protein interaction networks}.
\newblock Genome Biology \textbf{5} (2004)  R57

\bibitem{Chaouiya2011}
Chaouiya, C., Klaudel, H., Pommereau, F.:
\newblock {A modular, qualitative modeling of regulatory networks using Petri
  nets}.
\newblock In: {Modeling in systems biology: the Petri nets approach}. Springer
  (2011)  253--279

\bibitem{Rives2003}
Rives, A.W., Galitski, T.:
\newblock {Modular organization of cellular networks}.
\newblock Proceedings of the National Academy of Sciences of the USA
  \textbf{100} (2003)  1128--1133

\bibitem{Spirin2003}
Spirin, V., Mirny, L.A.:
\newblock {Protein complexes and functional modules in molecular networks}.
\newblock Proceedings of the National Academy of Sciences of the USA
  \textbf{100} (2003)  12123--12128

\bibitem{Milo2002}
Milo, R., Shen-Orr, S., Itzkovitz, S.,  et~al.:
\newblock {Network motifs: simple building blocks of complex networks}.
\newblock Science \textbf{298} (2002)  824--827

\bibitem{Alon2003}
Alon, U.:
\newblock {Biological networks: The tinkerer as an engineer}.
\newblock Science \textbf{301} (2003)  1866--1867

\bibitem{Bar2003}
Bar-Joseph, Z., Gerber, G.K., Lee, T.I.,  et~al.:
\newblock {Computational discovery of gene modules and regulatory networks}.
\newblock Nature biotechnology \textbf{21} (2003)  1337--1342

\bibitem{Segal2003}
Segal, E., Shapira, M., Regev, A.,  et~al.:
\newblock {Module networks: discovering regulatory modules and their condition
  specific regulators from gene expression data}.
\newblock Nature genetics \textbf{34} (2003)  166--176

\bibitem{Thieffry1999}
Thieffry, D., Romero, D.:
\newblock {The modularity of biological regulatory networks}.
\newblock Biosystems \textbf{50} (1999)  49--59

\bibitem{Han2008}
Han, J.D.J.:
\newblock {Understanding biological functions through molecular networks}.
\newblock Cell Research \textbf{18} (2008)  224--237

\bibitem{Siebert2009}
Siebert, H.:
\newblock {Dynamical and structural modularity of discrete regulatory
  networks}.
\newblock In: Proceedings of CompMod. Volume~6 of Electronic Proceedings in
  Theoretical Computer Science., Open Publishing Association (2009)  109--124

\bibitem{Delaplace2010}
Delaplace, F., Klaudel, H., Cartier-Michaud, A.:
\newblock {Discrete causal model view of biological networks}.
\newblock In: Proceedings of CMSB, ACM (2010)  4--13

\bibitem{Demongeot2010}
Demongeot, J., Goles, E., Morvan, M.,  et~al.:
\newblock {Attraction basins as gauges of robustness against boundary
  conditions in biological complex systems}.
\newblock PLoS One \textbf{5} (2010)  e11793

\bibitem{Thomas1973}
Thomas, R.:
\newblock {Boolean formalisation of genetic control circuits}.
\newblock Journal of Theoretical Biology \textbf{42} (1973)  563--585

\bibitem{Thomas1981}
Thomas, R.:
\newblock {On the relation between the logical structure of systems and their
  ability to generate multiple steady states or sustained oscillations}.
\newblock In: Numerical methods in the study of critical phenomena. Volume~9 of
  Springer Series in Synergetics.
\newblock Springer (1981)  180--193

\bibitem{Delbruck1949}
Delbr{\"u}ck, M.:
\newblock {Génétique du bactériophage}.
\newblock In: Unités biologiques douées de continuité génétique. Volume~8 of
  Colloques internationaux du CNRS. (1949)  91--103

\bibitem{Thieffry1995}
Thieffry, D., Thomas, R.:
\newblock {Dynamical behaviour of biological regulatory networks--II. Immunity
  control in bacteriophage lambda}.
\newblock Bulletin of Mathematical Biology \textbf{57} (1995)  277--297

\bibitem{Purnick2009}
Purnick, P.E.M., Weiss, R.:
\newblock {The second wave of synthetic biology: from modules to systems}.
\newblock Nature Reviews Molecular Cell Biology \textbf{10} (2009)  410--422

\end{thebibliography}

\end{document}